\documentclass{article}

\usepackage{arxiv}

\usepackage[utf8]{inputenc} 
\usepackage[T1]{fontenc}    
\usepackage{hyperref}       
\usepackage{url}            
\usepackage{booktabs}       
\usepackage{amsfonts}       
\usepackage{nicefrac}       
\usepackage{microtype}      
\usepackage{lipsum}
\usepackage{graphicx}
\graphicspath{ {./images/} }

\usepackage{amssymb}
\usepackage{amsmath}
\usepackage[T1]{fontenc}

\newtheorem{proposition}{Proposition}

\newtheorem{assumption}{Assumption}
\newtheorem{remark}{Remark}
\newtheorem{lemma}{Lemma}
\newenvironment{proof}[1][Proof: ] {\textit{#1}}{}
\newcommand{\N}{\mathbb{N}}
\newcommand{\R}{\mathbb{R}}
\newcommand{\E}{\mathbb{E}}
\newcommand{\1}{1\! \mathrm{l}}

\newcommand{\Pro}{\mathbb{P}}
\newcommand{\D}{\mathcal{D}}
\newcommand{\U}{\mathcal{U}}
\newcommand{\M}{\mathcal{M}}
\newcommand{\Po}{\mathcal{P}}

\usepackage{algorithm}
\usepackage{algpseudocode}
\usepackage{graphicx}
\usepackage{subcaption}

\title{Modeling Maximum drawdown Records with Piecewise Deterministic Markov Processe in Capital Markets}

\author{Rolando Rubilar-Torrealba\thanks{Departamento de Industrias, Universidad Técnica Federico Santa María, Valparaíso, Chile; rolando.rubilar@usm.cl}  \quad  \quad \quad  Lisandro Fermin \thanks{CIMFAV, Instituto de Ingeniería Matemática, Universidad de Valpara\'iso, Chile; Modal’X UMR CNRS 9023, Université Paris Nanterre, F92000 Nanterre, France. Email: lisandro.fermin@uv.cl; ljfermin@parisnanterre.fr}
\quad \quad \quad Soledad Torres\thanks{CIMFAV, Instituto de Ingeniería Matemática, Universidad de Valpara\'iso, Chile; Universidad Central de Chile, 8370242 Santiago, Chile . Email:  soledad.torres@uv.cl }
}
\date{December 3, 2024}

\begin{document}
\maketitle
\begin{abstract}
We propose to model the records of the maximum Drawdown in capital markets by means a Piecewise Deterministic Markov Process (PDMP). We derive statistical results such as the mean and variance that describes the sequence of maximum Drawdown records. In addition, we developed a simulation study  and techniques for estimating the parameters governing the stochastic process, using a practical example in the capital market to illustrate the procedure.
\end{abstract}


\section{Introduction}
\label{sec1}

One of the financial risk management strategies is to regulate the maximum losses that can occur over a fixed time horizon or maximum drawdown. As the level of risk is dynamic, the incentive structure of managers also changes in the face of the new maximum drawdown (\cite{pospisil2010portfolio}), leading to changes in investment decisions.

The management of financial risk by fund managers typically involves the application of a set of rules defining the behaviour of investments in the context of events that shape financial market prices (cited in references such as \cite{leal2005maximum, zhang2009effect, power2004risk}). Such events can be conditioned by observed losses in the capital markets. 

The strategies associated with the maximum Drawdown of a financial asset as a basic element for risk management generate a direct link between the price structure of financial assets with the statistical theory of extreme values, more specifically with the theory of records (\cite{embrechts2013modelling, gomes2015extreme, landriault2021analysis, li2022parisian}), assuming that the financial information of the financial asset is public knowledge for all members of the capital market.

In this article we link Piecewise Deterministic Markov Process (PDMP) processes with record theory for the particular case of a risk management process known as maximum Drawdown, which is widely used in the financial literature, but can be easily extended to other disciplines. This research aims to carry out a statistical and mathematical analysis of the characteristics of the limit distribution of the process that defines the maximum Drawdown records, based on the PDMP methodology, and to propose methods for estimating the parameters that determine the process.

The PDMP are processes which was originally developed by \cite{DAV} and have had an important relevance in determining phenomena in different areas of science. These processes are characterised by having a certain number of states that change randomly and by having a deterministic part that evolves in each of the states of the process.

The general characteristics of PDMP processes allow to approach areas that have not been deeply developed in the literature and to generate new research opportunities. In particular, record theory allows us to understand processes that evolve over time and to understand the asymptotic behaviour of certain phenomena, where record is the largest (smallest) value in a sequence of values that are observed over time.  

Among recent developments, we can cite \cite{rudnicki2015piecewise} who uses PDMP techniques to set up an analysis framework for biological systems, cell life cycle and other biological applications. Another interesting development is by \cite{kouretas2006parameter} who model biological networks for the study of antibiotic production. Similarly we can see developments in other areas of science such as physics, economics, among others (see \cite{lin2018efficient, schal1998piecewise, schmidli2010conditional}).

The results of the paper show a description of the analytical mean and variance of records obtained from the PDMP process in the context of the maximum Drawdown approach. In addition, simulation techniques and estimation of the parameters of the processes that describes the model were developed, providing important tools for potential applications in different areas of science.

The use of a variety of analytical tools allows for a comprehensive understanding of the stochastic process that governs the sequence of jumps that define the maximum drawdown. Markov chains are a fundamental tool for analysing the evolution of financial series, which have been extensively employed in economic and financial literature (\cite{mendoza2016multivariate,cui2018single, cui2019general}).

This paper is organized as follows. Section 1 presents the Introduction, the Model and preliminaries is presented in the Section 2. The Characterisation of the process is show in Section 3. The Section 4 corresponds to the Simulation and estimation of the process. In Section 5 we apply the methodology developed for an applied case, and the Section 6 corresponds to the Conclusions of this research.

\section{Model and preliminaries}\label{section 2}


The classical theory of records consider a time series $X_0, X_1, \ldots, X_n$ of random variables that can reach a maximum (minimum) value in a period of time. We define an upper and lower record as
$$X_n > \text{max}\left\lbrace X_0, X_1, \ldots, X_{n-1} \right\rbrace, $$
$$X_n < \text{min}\left\lbrace X_0, X_1, \ldots, X_{n-1} \right\rbrace. $$

Maximum \textit{Drawdown} in $T$ time can be informally defined as the largest drop from a peak of the time series to the smallest value of a valley, which implies reaching a new record in that time. Following \cite{magdon2004maximum}, we define the maximum \textit{Drawdown} process as a process driven by a Brownian process that indicates the maximum drop observed in a period in $[0, T]$. 
 
Let $X(t)$, $0 \leq t \leq T$, a Brownian motion with drift given by $X(t)=\sigma W(t) + \mu t$, where $\mu \in \mathbb{R}$ is the trend and $\sigma \geq 0$ is the diffusion parameter. 

The mathematical definition of maximum \textit{Drawdown} ($D$) is 
\begin{equation}\label{ec: Drawdown}
D(T; \mu, \sigma)=\sup_{t \in [0,T]} \left[ \sup_{s \in [0,t]} X(s) - X(t) \right], 
\end{equation}
with $0 \leq s \leq t \leq T$. This definition allows to characterise the maximum drop of a time series when considering a fixed value of $T>0$.

However, we are interested in the process for each time instant over the \textit{records} that it can reach at the maximum \textit{Drawdown}. For this stochastic process we consider the following definition of the occurrence times of a records in the process, such that
\begin{equation}\label{eq: proceso Draw}
T_{i} = \text{inf} \lbrace t>0: D(t) > D(T_{i-1}) \wedge \exists s>t, D(s)-D(t)=0\rbrace, 
\end{equation}
with $T_0 = 0$ and $i \in \mathbb{N}$, where $T_i$ represents the ocurrence of the $i-th$ record and $D(\cdot)$ corresponds to the maximum drawdown random variable defined in \eqref{ec: Drawdown}, which present a new record in the time series and, therefore, $D(T_i)$ and $T_i$ are random variables that define a stochastic process of drawdown records that must be characterised in terms of their time evolution and distribution.


The observation of the occurrence of new records allows modeling a stochastic process that alternates in different states and thus observe the long-term behaviour of the records that are achieved. In the following subsection a model is developed that allows to characterise the evolution of the records as a stochastic process.

\subsection{Piecewise Deterministic Markov Process}

In our model we consider a Piecewise Deterministic Markov Process (PDMP) as the main tool that allows us to characterise the stochastic process. PDMP processes were first introduced by \cite{DAV} and which has been popularized in recent years \cite{azais2018statistical}, becoming a powerful tool for working with stochastic differential equations, specifically with processes that can change states. PDMP processes is defined as follows

\subsubsection{Markov Process}

Let $( X_t, t \geq 0 )$ a Markov process with its continuous realizations on the right-hand side and with a limit on the left-hand side (\textit{c\'{a}dl\'{a}g}) almost surely. Moreover, the process $X$ has values in an open subset $\mathcal{X}\subseteq \mathbb{R}^d$, for $d \geq 1$. 

\begin{remark}
We use the notation $\partial \mathcal{X}$ as a frontier of $\mathcal{X}$ and $\overline{\mathcal{X}}$ the closure of $\mathcal{X}$.
For a Markov process $X$, we define the Markov semigroup as a family of operators $P_t$ acting on bounded measurable functions $f$ such that, for all $t > 0$ 
\begin{equation*}
P_t f(x) = \mathbb{E} \left[ f(X_t) \vert X_0 = x  \right],
\end{equation*}
and the infinitesimal generator $\mathcal{U}$ of $X$, acting on the functions $f$ as
\begin{equation}\label{ec: OperadorInf}
\mathcal{U} f(x) = \lim_{t \rightarrow 0} \frac{P_t f(x) - f(x)}{t}.
\end{equation}

\end{remark}
   
The operator $\mathcal{U}$ characterises the dynamics of the process and can be interpreted as the derivative in time of the semigroup $\partial_t P_t = \mathcal{U}$, at $t=0$. 

The PDMP process is determined by three characteristic components: 
\begin{enumerate}
\item Deterministic part of the movement between jumps $\Upsilon$.
\item $\lambda$ jump rate associated with a probability distribution.
\item Transition measure or $Q$ jump kernel.
\end{enumerate}


In addition, the transition measure $Q$ has the following characteristics:

\begin{enumerate}
\item For each fix $A\in \mathcal{E}$, mapping $x \rightarrow Q(A;x)$ is measurable.
\item $Q(\left\lbrace x  \right\rbrace; x )=0$
\end{enumerate}


  
 








Let us consider $(J_n)_{n\in \N}$ a time-homogeneous irreducible Markov chain taking values in the state space $K=\{1, \ldots, k\}$ with initial law 
$\pi_i=\Pro(J_0=i)$ for all $i\in K$ and  transition probability matrix $Q= (q_{ij})_{i,j 
\in K}$, i.e.
$$\Pro(J_{n+1}=j |J_n =i )= q_{ij}.$$
which implies that all states are connected and each state can be reached by a given state $i$.

We denote by $(T_n)_{n\in \N}$ the sequence of the record ocurrence random times and $(S_n)_{n\in \N}$ the random interval times; i.e. $S_n = T_{n+1} -T_{n}$.

In addition, we define $(\nu_t)_{t\in \R_+}$ by
$$\nu_t= \sum_{n\geq 0} J_n \1_{[T_n,T_{n+1}[} (t).$$

If $(J_n)_{n\in \N}$ is an irredicible Markov Chain, then  $(\nu_t)_{t\in \R}$ is an irreducible Markov process with the same initial law that your immersed Markov chain $(J_n)_{n\in \N}$. Here, $J_n$ is the state take by $\nu_t$ on $[T_n,T_{n+1}[$.

The matrix generator $A= (a_{ij})_{i,j\in K}$ of the process $(\nu_t)_{t\in \R_+}$ is given by 
\begin{equation}
 a_{ij} = \left\{
\begin{array}{cll}
\lambda_i q_{ij} & for & i\neq j\\
-\lambda_i(1-q_{ii})& for & i=j
\end{array}.
\right.
\end{equation}

The generator $A$ is stable and conservative; {\it i.e.} $\sum_{j\in K} a_{ij}=0$ and $a_{ij}\geq 0$ for $i\neq j$. which implies that the solutions of the system are driven by this differential operator which does not change over time and maintains stable solutions.

With these preliminaries we can define the record stochastic process $(R_t)_{t\in \R_+}$, which take values on $[0, 1]$ and suppose that $\Pro(R_0=r)=1$. 

\subsection{Definition of record process}

We consider a PDMP process that is modeled by mean of the Markov chain $(J_n)_{n\in \N}$. This process makes it possible to observe the expected value of maximum drawdown over time, allowing the development of applications that can be used in several areas of knowledge that require risk management in their analysis. The process  considering the following assumptions:

\begin{assumption}\label{assump.1}
\hspace{-0.1cm}
\begin{enumerate}
\item[i)] Between  two consecutive jump times $T_n$ and $T_{n+1}$, the continuous time process $(R_t)$ is constant and equal to $r_n = R_{T_n}$.
\item[ii)] The record jumps size $\Delta_{n}$ at the time $T_n$, is given by $\Delta_n= \rho(1-r_n)$, where $\rho$ is a random variable that takes values on $[0,1]$ with probability distribution $G_i$, for $J_n=i\in K$. The probability distributions $(G_i)_{i\in K}$ could be different.
\item[iii)] The interval time $S_n$ is a random variable with exponential distribution of parameter $\lambda_{i}$, for  $i \in K$, where the jump rate $\lambda_i$ of state $i$ is a strictly positive constant.
\end{enumerate}
\end{assumption}


Note that he records process increases by the value $\Delta_n$  at the time $T_n$.  After that, the record process remains constant until the next jump $T_{n+1}$. 
Thus, the sample path of the stochastic process $(R_t,\nu_t)_{t \in \R_+}$ with values in $[0,1] \times K$ starting from a fixed point $(r,\nu)$ is defined in the following way, as we illustrate in Figure \ref{fig1}.

\begin{figure}[ht]
\centering
{\includegraphics[width=0.9\textwidth]{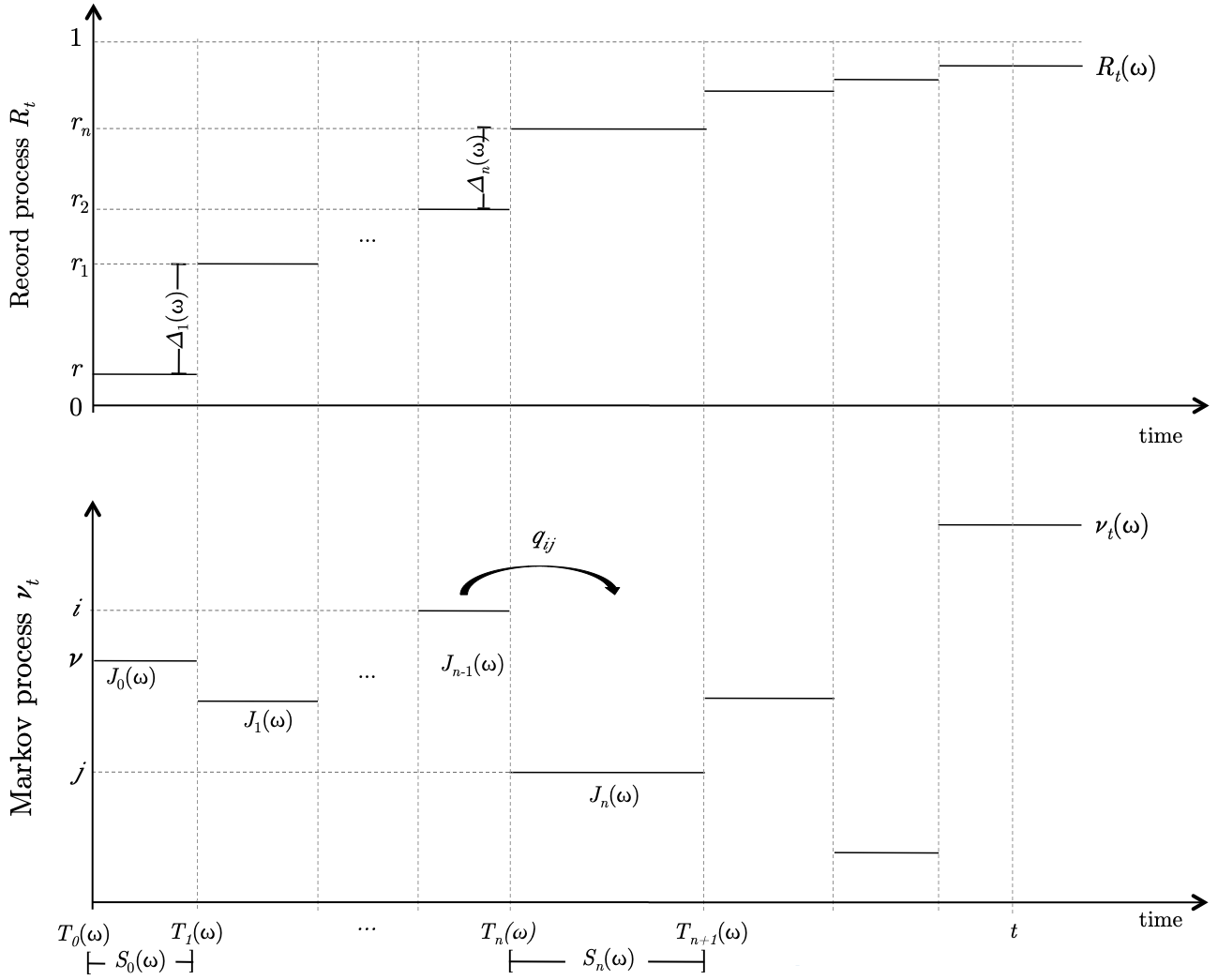}}
\caption{Sample path of the PDMP process $(R_t, \nu_t)$.} \label{fig1}
\end{figure}

First, say that $\nu_t=J_0=\nu$ for $t< T_1= S_0$, where $S_0$ stands for the first jump time of $\nu_t$, which has an exponential distribution of parameter $\lambda_{J_0}$, and $\nu_{T_1}=J_1$. Now, we define $\rho_{J_1}$ as a random variable with density distribution $g_{J_1}$, and then we take $\Delta_1=\rho_{J_1}(1-r)$. Then the sample path $(R_t)$ up to the first jump time $T_1$ is now defined as follows:
\begin{equation*}
\begin{array}{llll}
R_t        &=&  r                      & \text{ if } 0\leq t< T_1,\\
R_{T_1} &=& r +  \Delta_1& .
\end{array}
\end{equation*}

The process now restarts from $r_1= R_{T_1}$ according to the same recipe. Thus, we define $S_1$ a random variable with exponential distribution of parameter $\lambda_{J_1}$, so we take $T_2= T_1 + S_1$. The Markov process $(\nu_t)$ jump to regime $J_2$ with rate $q_{J_1J_2}$, thus  $\nu_{T_2}=J_2$ and $\Delta_2=\rho_{J_2}(1-r_1)$ where  $\rho_{J_2}$ is a random variable with probability distribution $G_{J_2}$. Then, the sample path $(R_t)$ up to the second jump time, starting from $r_1$ at time $T_1$, is defined as
\begin{equation*}
\begin{array}{llll}
R_t        &=& r_1                       &\text{ if } T_1\leq t < T_2,\\
R_{T_2} &=& r_1 +  \Delta_2&;
\end{array}
\end{equation*}
and so on. Finally, for all $n\in\N$ and for $k=1,\ldots, n$, we take $(J_k)_{k=0:n}$ a sample path of the Markov chain, $S_k$ a random variable with exponential distribution of parameter $\lambda_{J_k}$, $T_{k+1}=T_k+S_k$, $\nu_{T_{k+1}}=J_{k+1}$, and  we define $\Delta_{k+1}= \rho_{J_{k+1}}(1-r_{k})$ where $\rho_{J_{k+1}}$ is a random variable with probability distribution $G_{J_{k+1}}$ and $r_k=R_{T_{k}}$. Then, we have
\begin{equation*}
\begin{array}{llll}
R_t  &=&  r + \sum_{k=1}^n \Delta_k & \text{ if } T_n\leq t < T_{n+1},
\end{array}
\end{equation*}
which can be rewritten as
\begin{equation}\label{eq.1.2}
R_t =  r  + \sum_{k\geq 1} \Delta_k  \1_{(t\geq T_k)}.
\end{equation}

Note that the record jump size $\Delta_k$ is a funtion of $\nu_{T_k}=J_k$ and $R_{T_{k-1}}=r_{k-1}$.

\begin{lemma}\label{lemma.1}
The records process $(R_t)_{t \in \R_+}$ defined in \eqref{eq.1.2} with initial value $R_0\in[0,1[$ satisfy
\begin{equation}
R_t= R_0+(1-R_0)R^0_t, \label{ProcessLemma1}
\end{equation}
where the process $(R^0_t)_{t \in \R_+}$ given by
\begin{equation}
R^0_t= \sum_{k\geq 1} \rho_{J_{k}}\left(\prod_{i=1}^{k-1}(1-\rho_{J_i}) \right)  \1_{(t\geq T_k)}, \label{ProcessLemma2}
\end{equation}
is the records process obtained when the initial condition is zero, $R_0^0=0$.
\end{lemma}

\begin{proof}
$R_t$ is given by equation \eqref{eq.1.2} where $\Delta_k=\rho_{J_{k}}(1-r_{k})$ are given by the following recursive equations:
\begin{eqnarray*}
\Delta_1 & = & \rho_{J_1}(1-R_0)\\
\Delta_2 & = & \rho_{J_2}(1-R_{T_1})= \rho_{J_2}(1-R_0-\Delta_1)=\rho_{J_2}(1-\rho_{J_1})(1-R_0)\\
\Delta_3 & = & \rho_{J_3}(1-R_{T_2})=  \rho_{J_3}(1-R_0 -\Delta_1 -\Delta_2)=\rho_{J_3}(1-\rho_{J_2})(1-\rho_{J_1})(1-R_0)\\
\vdots & & \\
\Delta_k& = & \rho_{J_k}(1-R_{T_{k-1}})=  \rho_{J_k}(1-R_0 -\Delta_1 -\ldots -\Delta_{k-1})=\rho_{J_k}\prod_{i=1}^{k-1}(1-\rho_{J_i})(1-R_0).\\
\end{eqnarray*}

Then, 
\begin{equation*}
R_t =  R_0  + (1-R_0) \sum_{k\geq 1} \rho_{J_{k}}\left(\prod_{i=1}^{k-1}(1-\rho_{J_i}) \right)  \1_{(t\geq T_k)}.
\end{equation*}
When $R_0=0$ we obtain the records process 
\begin{equation*}
R^0_t = \sum_{k\geq 1} \rho_{J_{k}}\left(\prod_{i=1}^{k-1}(1-\rho_{J_i}) \right)  \1_{(t\geq T_k)}.
\end{equation*}
\end{proof}

Considering that $R_0$ and $\nu_0$ are independents, we have that the process $(R_t,\nu_t)_{t\in \R_+}$ is a piecewise deterministic Markov process (PDMP).

We denote by $\M$ the set of measurable real valued functions on $E=[0,1]\times K$ and by $\M_0$ the set of bounded measurable real valued function on $E$.

From \cite{DAV}, we have that the domain $\D(\U)$ of the infinitesimal generator $\U$ of $(\nu_t,R_t)_{t\in \R_+}$ consists of those functions $f \in \mathcal{M}$ integrable with respect to the probability measures $G$. For $f \in \D(\U)$ the infinitesimal generator $\U$ is given by
\begin{equation}\label{eq.1.3}
\U f(r,\nu)= \lambda_{\nu} \sum_{j\in K} q_{\nu j} \int_{[0,1]} \left( f(r+ \rho(1-r), j) - f(r,\nu)\right)G_j(d\rho),
\end{equation}
with $(r,\nu) \in E$.

We denote by $P: (\nu,j,r,B,t) \rightarrow P_{\nu j}(r,B,t)$ the transition probability of $(R_t, \nu_t)$; i.e.
$$P_{\nu j}(r,B,t) = \Pro(R_t\in B, \nu_t=j | R_0=r, \nu_0=\nu).$$

It is defined for all $r\in[0,1]$, $\nu,j\in K$, $t\in \R_+$ and $B\in \mathcal{B}([0,1])$, with $\mathcal{B}([0,1])$ being the Borel $\sigma$-field of $[0,1]$.

The transition probability permits to give an expression for the probability distribution of $(R_t,\nu_t)$ in the following way:
\begin{equation}\label{eq.1.3.1}
\Pro(R_t\in B, \nu_t=j|R_0=r) = \sum_{\nu\in K} \pi_{\nu} P_{\nu j}(r,B,t).
\end{equation}

For $t \in R_+$ fixed, define an operator $\Po_t: \M_0 \rightarrow \M_0$ by the following conditional expectation given the starting point $(r,\nu)$
$$\Po_t f (r,\nu)= \E_{(r,\nu)}[f(R_t,\nu_t)]= \sum_{j\in K} \int_{[0,1]} f(r',j) P_{\nu j}(r,dr',t).$$
Here, $(\Po_t; t \in \R_+)$ is the semigroup associated with the infinitesimal generator $\U$; (i.e $\Po_{s+t}= \Po_s \Po_t$).

From \cite{DAV}, we recall that for $t$ fixed, $(r,\nu)\in E$  and $f \in \D(\U)$, $z(t,r,\nu)=\Po_t f (r,\nu)$  is the unique solution of the following partial differential equation (EDP):

\begin{equation}\label{eq.1.4}
\left\{
\begin{array}{rll}
\displaystyle{\frac{\partial z}{\partial t}} (t,r,\nu) &=& \U z(t,r,\nu),\\
z(0,r,\nu) &=& f(r,\nu).
\end{array}
\right.
\end{equation}

\section{Statistical properties of record process}\label{Variabilidad de R}

\subsection{Mean}

First, to analyze the variability of the process we proceed to calculate the mean of the stochastic process. The mean of $R_t$ given the starting point $R_0=r$ is given by

$$
m(t,r)= \sum_{\nu\in K} \pi_{\nu} m(t,r,\nu),
$$
where $\pi=(\pi_{\nu})_{\nu\in K}$ is the initial law of the Markov chain $(\nu_t)_{t\in\R_{+}}$ and $m(t,r,\nu)=  \E_{(r,\nu)}[R_t]$
is the conditional expectations of $R_t$, given the starting point $(r,\nu)$.
\vspace{0.1cm}

\begin{proposition}\label{prop.1.2}
The conditional expectations $m(t,r)= \E_{r}[R_t]$ of $R_t$, with the starting point $R_0=r$ is given by 
\begin{equation}\label{eq.1.10}
m(t,r) = \pi e^{B t} r + \pi \left( e^{B t} - I \right)B^{-1} \Lambda Q \mu, 
\end{equation}
where  $m(t,r)=(m(t,r,j))_{j\in K}$, $B= \Lambda Q (I-M)- \Lambda$, $\Lambda=diag(\lambda_{j}, j\in K)$, $Q$ the transition probability matrix of the Markov chain $(\nu_t)_{t\in R_{+}}$ and $\pi$ its initial distribution, $\mu=(\mu_j)_{j \in K}$ with $\mu_j=\int_{[0,1]} \rho G_j(d\rho)$, $M=diag(\mu_{j}, j\in K)$, and $I$ is the $k\times k$ identity matrix.
\end{proposition}

\begin{proof}
We have that $m(t,r,\nu)$ is the unique solution of the system \eqref{eq.1.4} for $f(r, \nu)= r$, then we can rewrite \eqref{eq.1.4} as
\begin{equation}\label{eq.1.11}
\left\{
\begin{array}{l}
\displaystyle{\frac{\partial m}{\partial t} (t,r,\nu) = \lambda_\nu \sum_{j\in K}  q_{\nu j} \int_{[0,1]} \left(m(t,r+\rho(1-r),j)-m(t,r,\nu) \right)G_j(d\rho),}\\
m(0,r,\nu) = r,
\end{array}
\right.
\end{equation}

On the other hand, from Lemma \ref{lemma.1}  
\begin{eqnarray}\label{eq.1.12}
m(t,r+\rho(1-r),j) & = & \E[R_t | R_0=r+\rho(1-r), \nu_0=j] \nonumber\\
& = &\E[\rho + (1-\rho)(r+(1-r)R^0_t) | R_0=r+\rho(1-r), \nu_0=j] \nonumber \\
& = &\E[r+\rho - \rho r+R_t^0 - r R_t^0 - \rho R_t^0 +  \rho r R_t^0 | R_0=r+\rho(1-r), \nu_0=j] \nonumber \\
& = &\E[r+(1-r)R_t^0-\rho (r+(1+r)R_t^0-1) | R_0=r+\rho(1-r), \nu_0=j] \nonumber \\
& = &\rho + (1-\rho)\E[R_t | R_0=r, \nu_0=j] \nonumber\\
& = &\rho + (1-\rho)m(t,r,j).
\end{eqnarray}
Then, from \eqref{eq.1.11} and \eqref{eq.1.12} we have
\begin{eqnarray}\label{eq.1.13}
\frac{\partial m}{\partial t} (t,r,\nu) &=& \lambda_\nu \sum_{j\in K}  q_{\nu j} \left( \mu_j + (1-\mu_j)m(t,r,j) - m(t,r,\nu)\right). \nonumber 
\end{eqnarray}

This differential equations system can be written as a system of equations in matricial terms
\begin{eqnarray}\label{eq.1.14}
\frac{\partial \bar{m}}{\partial t} (t,r) &=& \Lambda Q\mu + B \bar{m}(t,r).  
\end{eqnarray}

The solution of  this differential equations system is given by
\begin{equation}\label{eq.1.15}
\bar{m}(t,r) = e^{B t} r + \left( e^{B t} - I \right)B^{-1} \Lambda Q \mu.
\end{equation}

Thus, we obtain the result given in equation \eqref{eq.1.10} for $m(t,r)=\pi \bar{m}(t,r)$ .

\end{proof}

\subsubsection{Case of one-state}


Using equation \eqref{eq.1.12}, we have

\begin{equation}
\left\{
\begin{array}{l}
\displaystyle{\frac{\partial m}{\partial t} (t,r,\nu_0) =  \mu_{\nu_0} (1-m(t,r,\nu_0)),}\\
m(0,r,\nu_0) = r,
\end{array}
\right.
\end{equation}
with $\mu_{\nu_0}=\int_{[0,1]} \rho G_{\nu_0}(d\rho)$, and $\nu_0$ corresponds to the state $\nu_0$ which remains fixed for all time $t$.

Regrouping terms, we have
\begin{eqnarray}
\frac{\partial m}{\partial t} (t,r,\nu_0) + \mu_{\nu_0}m(t,r,\nu_0) = \mu_{\nu_0},  \nonumber 
\end{eqnarray}
solving the differential equation we have
\begin{eqnarray}\label{ec: primer momento 1-s}
m(t,r,\nu_0) = 1- \dfrac{1-r}{e^{\mu_{\nu_0} t}},   
\end{eqnarray}
providing the functional form of the process mean $R_t$ for the case of a one-state. We can analyse the function that characterises the process, and obtain the following results
\begin{equation*}
\lim_{t \rightarrow 0} m(t,r,\nu_0) = r,
\end{equation*}
\begin{equation}\label{eq: lim_mean}
\lim_{t \rightarrow \infty} m(t,r,\nu_0) = 1,
\end{equation}
describing the dynamics of the $R_t$ process over time.

\subsubsection{Case of two-state}


From \eqref{eq.1.14} for the two-state case, we have
\begin{eqnarray}\label{2-states}
\frac{\partial m}{\partial t} (t,r,j) &=&  a \  m(t,r,j) + b \ m(t,r,i) + e \nonumber\\
\frac{\partial m}{\partial t} (t,r,i) &=& c \ m(t,r,j) + d \ m(t,r,i) + f \\
m(0,r,\nu_0) &=& r, \nonumber
\end{eqnarray}
where the coefficients are defined as follows
\begin{eqnarray*}
a&=& \lambda_j q_{1,1} (1-\mu_j) - \lambda_j\nonumber\\
b&=& \lambda_j q_{1,2} (1-\mu_i)\nonumber\\
c&=& \lambda_i q_{2,1} (1-\mu_j)\nonumber\\
d&=& \lambda_i q_{2,2} (1-\mu_i) - \lambda_i\nonumber\\
e&=& \lambda_j (q_{1,1} \mu_j + q_{1,2} \mu_i)\nonumber\\
f&=& \lambda_i (q_{2,1} \mu_j + q_{2,2} \mu_i).\nonumber\\
\end{eqnarray*}
Here $\lambda_j, \lambda_i$ represent the jump rate from state $j$ and $i$, respectively; $\mu_j, \mu_i$ represent the expected value of the jump in state $j$ and $i$, respectively and  $q_{1,1}, q_{1,2}, q_{2,1}, q_{2,2}$  the components of the transition matrix $Q$ for two states. 

For solving numerically the system of differential equations \eqref{2-states}, the  Runge-Kutta method \cite{soetaert2010solving, zheng2017modeling} can be used. An example is given in Table  \ref{table: parameters}, where $X$ represents the random variable for the time between jump events subject to states $i$ and $j$ and $\rho$ corresponds to the random variable for the size of the jump. 

\begin{table}[ht] 
\centering
\begin{tabular}{lcl}
\cline{2-3}
& \multicolumn{2}{c}{Distribution}                                   \\
& $\nu_t= i$                        & \multicolumn{1}{c}{$\nu_t= j$} \\ \hline
\multicolumn{1}{r}{$X \sim$} & $Exponential(2)$                 & $Exponential(1)$               \\
\multicolumn{1}{r}{$\E[X] =$}     & \multicolumn{1}{l}{$0.5$} & $1.0$       \\         
\multicolumn{1}{r}{$\rho \sim$}     & \multicolumn{1}{l}{$Beta(2, 20)$} & $Beta(2, 30)$                  \\
\multicolumn{1}{r}{$\E[\rho]=$}     & \multicolumn{1}{l}{$0.012$} & $0.015$                  \\ \hline
\end{tabular}
\caption{Characteristics of the stochastic elements of the process $R_t$} \label{table: parameters}
\end{table}

\begin{equation}
Q= \left[
\begin{matrix}
q_{1,1} &  q_{1,2}\\
q_{2,1} &  q_{2,2}
\end{matrix}\right] = \left[
\begin{matrix}
0.6 &  0.4\\
0.5 &  0.5
\end{matrix}
\right],
\end{equation} \label{eq: transition_matrix}

In this case, the representation of the evolution of the mean for each state ($\nu$) by the Runge-Kutta method is described as

\begin{equation}
\left\{
\begin{array}{l}
\displaystyle{m(t+\Delta,r,\nu) =  m(t,r,\nu) + \frac{1}{6}(K_1+2K_2+2K_3+K_4),}\\
K_1 = \Delta f(t,m(t,r,\nu)), \\
K_2 = \Delta f(t+\dfrac{1}{2}\Delta,m(t,r,\nu)+\frac{1}{2}K_1), \\
K_3 = \Delta f(t+\dfrac{1}{2}\Delta,m(t,r,\nu)+\frac{1}{2}K_2), \\
K_4 = \Delta f(t+\Delta,m(t,r,\nu)+K_3), \\
\end{array}
\right.
\end{equation}

The following figure \ref{fig: mean path}, shows the evolution of the expected value of the process $R_t$, as a result of the numerically solution of the system of equations defined in \eqref{2-states} for the two-state case with initial value $r=0$. It is observed that the mean values increase rapidly at the beginning and then decrease with time, approaching the limit $1$. The figure is like the learning curve, which describes a situation in which the task may be easy to learn and the learning progression is initially rapid and fast.

The curve levels off after a certain instant $t$. In the graph it can be noticed that this instant $t \sim 30$ could be described as a plateau. 

\begin{figure}[h]
	\centering\includegraphics[width=10cm]{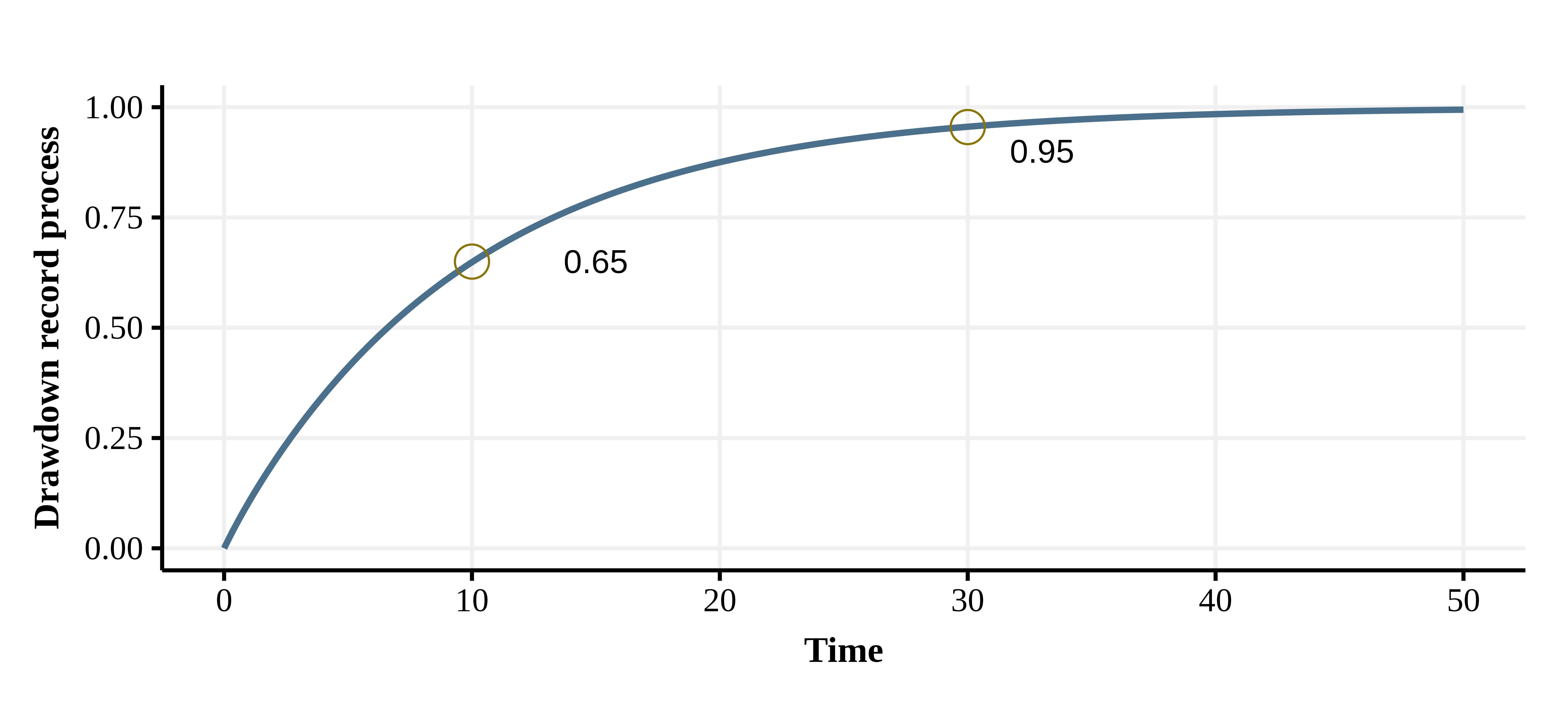}
	\caption{Evolution of the expected value of $R_t$ process.}\label{fig: mean path}
\end{figure}

\subsection{Variance estimation}

In this section, the variability of the record process $R_t$ is studied, starting from the computation of the second moment given the  initial point $R_0^2 = r^2$

$$
m_2(t,r)= \sum_{\nu\in K} \pi_{\nu} m_2(t,r,\nu),
$$
where $\pi=(\pi_{\nu})_{\nu\in K}$ is the initial law of the Markov chain $(\nu_t)_{t\in\R_{+}}$ and $m_2(t,r,\nu)=  E_{r,\nu}[R_t^2]$ is the conditional expectations of $R_t$, given the initial point $(r^2,\nu)$.

\begin{proposition}\label{prop.1.3}
	The variance of records process $(R_t)_{t \in \R_+}$ defined in \eqref{eq.1.2} with initial value $R_0^2 = r^2$  is, 
\begin{eqnarray}\label{ec: Varianza}
Var(R_t) &=& \E_{r,\nu} \left[ R_t^2 \right]-\E_{r,\nu}^2 \left[R_t\right] \nonumber \\
&=&  \pi e^{Ht} r^2 + \pi \int_0^t e^{-H(t-s)} \Lambda Q \mu_2 ds + \pi \int_0^t e^{H(t-s)} K \overline{m}(s,r)ds \nonumber \\
& - &  \left( \pi e^{B t} r + \pi \left( e^{B t} - I \right)B^{-1} \Lambda Q \mu \right)^2
\end{eqnarray}
where  $B= \Lambda Q (I-M)- \Lambda$, $\Lambda=diag(\lambda_{j}, j\in K)$, $Q$ the transition probability matrix of the Markov chain $(\nu_t)_{t\in R_{+}}$, $H = [\lambda Q (I-2M + M_2 - \Lambda)]$, $K= 2 \Lambda Q (M-M_2)$,  $\mu=(\mu_j)_{j \in K}$ with $\mu_j=\int_{[0,1]} \rho G_j(d\rho)$, $\mu_2=(\mu_{2,j})_{j \in K}$ with $\mu_{2,j}=\int_{[0,1]} \rho^2 G_j(d\rho)$.

\end{proposition}
\vspace{0.1cm}

\begin{proof}
Let $f(r,\nu)=r^2$, $m_2(t,r,\nu) =\E_{r,\nu}[R_t^2]$.  The unique solution of the system \eqref{eq.1.4} that allow to define the behaviour of the second moment in the stochastic process $R_t$, considering a finite number of  of states is given by
\begin{equation}\label{ec: SistemaM2}
\left\{
\begin{array}{l}
\displaystyle{\frac{\partial m_2}{\partial t} (t,r,\nu) = \lambda_\nu \sum_{j\in K}  q_{\nu j} \int_{[0,1]} \left(m_2(t,r+\rho(1-r),j)-m_2(t,r,\nu) \right)G_j(d\rho),}\\
m_2(0,r,\nu) = r^2,
\end{array}
\right.
\end{equation}
From Lemma 1, we can write the second moment as
\begin{eqnarray}\label{ec: M2a}
m_2(t,r_1 + \rho (1-r) , j) &=& \E [R_t^2| R_0 = r+\rho (1-r), \nu_0 = j ] \nonumber \\
& = & \rho^2 + 2\rho (1-\rho) m(t,r,j) +m_2 (t,r,j)(1-\rho)^2.
\end{eqnarray}


Let  $\mu_j = \int_{[0,1]} \rho G_j (d\rho)$, $\mu_{2,j} = \int_{[0,1]} \rho^2 G_j (d\rho)$. Combining \eqref{ec: SistemaM2} y \eqref{ec: M2a}, we have
\begin{eqnarray*}\label{ec: dM2}
\dfrac{\partial m_2}{\partial t}(t,r,\nu) &=& \sum_{j=1}^k \lambda_{\nu} q_{\nu,j} (\mu_{2,j} + 2(\mu_j - \mu_{2,j}) m(t,r,j) \nonumber\\ 
& + &  (1- 2 \mu_j + \mu_{2,j})m_2(t,r,j) - m_2(t,r,\nu)),
\end{eqnarray*}
which in matrix form is written as
\begin{eqnarray}\label{ec: dM2Mat}
\dfrac{\partial \overline{m}_2}{\partial t} (t,r) & = & \Lambda Q \mu_2 + 2 \Lambda Q (M-M_2) \overline{m}(t,r)\nonumber\\
& + & \Lambda Q (I -2 M +M_2) \overline{m}_2(t,r)  -  \Lambda \overline{m}_2 (t,r).  
\end{eqnarray}

Let $H = [\lambda Q (I-2M + M_2 - \Lambda)]$ y $K= 2 \Lambda Q (M-M_2)$.  Regrouping terms from \eqref{ec: dM2Mat}, we get the following differential equation
\begin{equation}
\dfrac{\partial \overline{m}_2}{\partial t}(t,r) - H \overline{m}_2(t,r) = \Lambda Q \mu_2 K \overline{m}(t,r),
\end{equation}
which has as a solution
\begin{eqnarray}\label{ec: M2sol}
\overline{m}_2(t,r) = e^{Ht} r^2 + \int_0^t e^{-H(t-s)} \Lambda Q \mu_2 ds + \int_0^t e^{H(t-s)} K \overline{m}(s,r)ds. 
\end{eqnarray}
Finally, from  \eqref{eq.1.10} and  \eqref{ec: M2sol}, we obtain
\begin{eqnarray}
Var(R_t) &=&  e^{Ht} r^2 + \int_0^t e^{-H(t-s)} \Lambda Q \mu_2 ds + \int_0^t e^{H(t-s)} K \overline{m}(s,r)ds \nonumber \\
& - & \left( e^{B t} r + \left( e^{B t} - I \right)B^{-1} \Lambda Q \mu \right)^2 \nonumber.
\end{eqnarray}
\end{proof}

\subsubsection{Case of one-state}

In this section  the case of a steady state such that $\nu= \nu_0$ is presented, which shows the intuition of the results obtained for the n-state case. From equation \eqref{ec: SistemaM2} we have
\begin{eqnarray}
\dfrac{\partial \overline{m}_2}{\partial t} (t,r,\nu_0) & = & \lambda \mu_2 + 2 \lambda (\mu-\mu_2) m(t,r,\nu_0) +  \lambda ( \mu_2 - 2 \mu) m_2(t,r,\nu_0)\nonumber\\
m_2(0,r,\nu_0) &=& r^2,
\end{eqnarray}
which is equivalent to 
\begin{eqnarray}\label{eq: reagrupado}
\frac{\partial m_2}{\partial t} (t,r,\nu_0) + c \ m_2(t,r,\nu_0) = a+b\  m(t,r,\nu_0), 
\end{eqnarray}
with $a= \lambda \mu_2$, $b= 2 \lambda (\mu-\mu_2)$ and $c= \lambda (2 \mu -\mu_2)$.
Combining equations \eqref{ec: primer momento 1-s} and \eqref{eq: reagrupado}, we have
\begin{eqnarray}
m_2 (t,r,\nu_0) &=&  \dfrac{a+b}{c} + \dfrac{r-1}{(c-\mu)e^{c t}}   +  \left( r^2 - \dfrac{a+b}{c} - \dfrac{r-1}{(c-\mu)} \right) \frac{1}{e^{ct}}.
\end{eqnarray}
Therefore, the variance of the process $R_t$ is given by
\begin{eqnarray}\label{eq: lim_var}
Var(R_t) =  \nonumber \\
&& \hspace{-2cm} \dfrac{a+b}{c} + \dfrac{r-1}{(c-\mu)e^{c t}}  +  \left( r^2 - \dfrac{a+b}{c} - \dfrac{r-1}{(c-\mu)} \right) \frac{1}{e^{ct}}- \left( 1- \dfrac{1-r}{e^{\mu_{\nu_0} t}}  \right)^2 \nonumber \\
= &&  \dfrac{r-1}{(c-\mu)e^{c t}}  +  \left( r^2 - 1 - \dfrac{r-1}{(c-\mu)} \right) \frac{1}{e^{ct}} +2\dfrac{1-r}{e^{\mu_{\nu_0} t}} - \left(\dfrac{1-r}{e^{\mu_{\nu_0} t}}  \right)^2 \nonumber \\
\leq &&  \frac{2(1-r)}{e^{\mu_{\nu_0} t}},
\end{eqnarray}
providing the functional form of the variance of process  $R_t$ for the case of one-state. Analysing the function that characterises the second moment of the process, and obtain the following results 

\begin{equation*}
\lim_{t \rightarrow 0} m_2(t,r,\nu_0) = r^2,
\end{equation*}
\begin{equation}
\lim_{t \rightarrow \infty} m_2(t,r,\nu_0) = 1,
\end{equation}
describing the dynamics of the second moment of process over time, with a variance of

\begin{equation}\label{eq: var_convergence}
Var(R_t) \leq 2(1-r) e^{-\mu_{\nu_0} t},
\end{equation}

\subsubsection{Case of two-state}

In this section the variance process of $R_t$ for the two-state case is computed, which can be extended to $n$ different states.  From \eqref{ec: dM2Mat} we have
\begin{eqnarray*}
\frac{\partial m_2}{\partial t} (t,r,j) &=&  a \  m_2(t,r,j) + b \ m_2(t,r,i) + c \  m_2(t,r,j) + d +e \ m(t,r,j) +f \ m(t,r.i) \nonumber\\
\frac{\partial m_2}{\partial t} (t,r,i) &=& g \ m_2(t,r,j) + h \ m_2(t,r,i) + k \  m_2(t,r,i) + l + n \ m(t,r,j) + p \ m(t,r,i) \nonumber\\
m_2(0,r,\nu_0) &=& r, \nonumber
\end{eqnarray*}
where the coefficients are defined as follows
\begin{eqnarray*}
a&=& \lambda_j q_{1,1} (1-2\mu_j + \mu_{2,j}) \nonumber\\
b&=& \lambda_j q_{1,2} (1-2\mu_i + \mu_{2,i})\nonumber\\
c&=& \lambda_j  \nonumber\\
d&=& \lambda_j (q_{1,1} \mu_{2,j}+ q_{1,2} \mu_{2,i}) \nonumber\\
e&=&  2\lambda_j q_{1,1} (\mu_j-\mu_{2,j}) \nonumber\\
f&=&  2\lambda_1 q_{1,2} (\mu_i-\mu_{2,i}) \nonumber\\
g&=&  \lambda_i q_{2,1} (1-2 \mu_j+\mu_{2,j}) \nonumber\\
h&=& \lambda_i q_{2,2} (1-2 \mu_i+\mu_{2,i}) \nonumber\\
k&=& \lambda_i \nonumber\\
l &=& \lambda_i (q_{2,1} \mu_{2,j}+ q_{2,2} \mu_{2,i}) \nonumber\\
n &=& 2 \lambda_i q_{2,1} (\mu_j-\mu_{2,j}) \nonumber\\
p &=& 2 \lambda_i q_{2,2} (\mu_i-\mu_{2,i}) .\nonumber\\
\end{eqnarray*}
Here, $\lambda_j, \lambda_i$ represent the jump rate from state $j$ and $i$, respectively; $\mu_j, \mu_i$  the expected value of the jump in state $j$ and $i$, respectively; $q_{1,1}, q_{1,2}, q_{2,1}, q_{2,2}$  the components of the transition matrix $Q$ for two states. 
In this case, the system of differential equations can be solved by Runge-Kutta method \cite{soetaert2010solving}, using parameters defined in the Table \ref{table: parameters}. In Figure \ref{fig: variance path} we plot the evolution of the variance of the process $R_t$,  with initial value $r^2=0$, and demonstrate convergence to zero.
As illustrated in the graphics,  the maximum variance is obtained around the time $t= 5$, which decays rapidly, approaching its asymptotic value equal to zero at period $t=50$. 

\begin{figure}[h]
\centering\includegraphics[width=14cm]{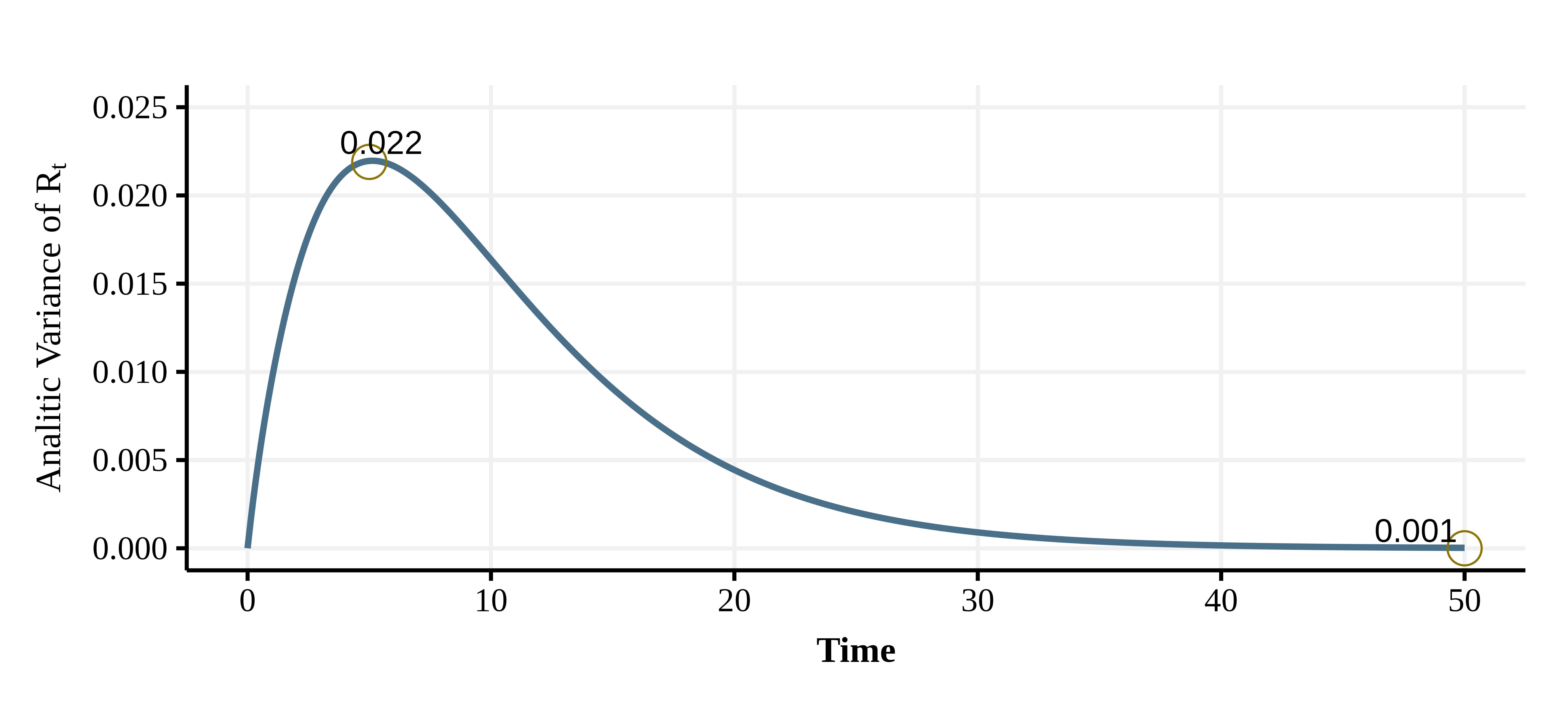}
    \caption{Evolution of variance of $R_t$ process.}\label{fig: variance path}
\end{figure}


\begin{remark}
From Chebyshev inequality (\cite{saw1984chebyshev}) 
\begin{equation}
\Pro( |R_t - m(t,r) | \geq t^{-\alpha}) \leq t^{2\alpha} Var(R_t),
\end{equation}
for all $\alpha>0$. As $t$ tends to infinity, equations \eqref{ec: primer momento 1-s}, \eqref{eq: lim_var}, and \eqref{eq: var_convergence} give
\begin{equation}
\Pro( |R_t - 1 | \geq t^{-\alpha})  \leq 2(1-r) e^{-\mu_{\nu_0} t} t^{2\alpha} = 0,
\end{equation}

%

%
%
and therefore, the process $R_t$ converge in probability.
\end{remark}
%



\section{Simulation and estimation in the record process}

In this section we proceed to perform a series of simulations in order to characterise the record phenomenon from an empirical perspective. Secondly, we present the estimation method of the  process $R_t$ for different states of the time series.  This study allow us to have a concrete tool to generate applications of the stochastic $R_t$ process.

\subsection{Simulation of record process}

For  simulate the process $R_t$ we will use the approximation made by Thomas \cite{thomas2019stochastic} for PDMP processes.  
We will use the theoretical development developed in  previous section for  two-state case to check the simulation against the analytical results. Through this application, we  present a methodology that can be used in various areas of science, such as economy (see \cite{chang2009macroeconomic, liu2011evolving}), pharmacology (see \cite{Lisandro}), finance (see, \cite{liu2012stock}). 

First, we must characterise the changes of state of the process $R_t$ for each of the jumps. Each jump of the process is characterised by a time occurrence and a jump size. Both characteristics are modelled independently by defining a transition matrix $Q$.

The algorithm is described below for the case when the  process $R_t$ has  two-state, and  the times between jumps has an exponential distribution with parameters $\lambda_i$ and $\lambda_j$. When a new record is reached, the size of the jump has a Beta distribution depending on whether they belong to the $\nu_i$ or $\nu_j$ state, respectively. The parameters of the corresponding Beta distributons are ($\alpha_i$, $\alpha_j$) and ($\beta_i$, $\beta_j$).
The extension to n-states is natural  by  adding sections referring to the  states.
\begin{algorithm}
\caption{Simulation of the $R_t$ process for two states}\label{sim: proc}
\begin{algorithmic}
\Require Set the initial values of the process $R_0=r$ and the initial state $\nu_0$. Define the probabilities of the states of the transition matrix $Q$. Additionally it sets the time counter $k$ and inicial time $t$ to zero.
\While{$t \leq T$}
\If{$\nu_t = i$}
    \State Simulate $X \sim Exponential(\lambda_i) $
    \State Simulate $Y \sim Beta(\alpha_i, \beta_i)$
    \State Simulate $u \sim Uniform(0,1)$
    \If{$Q(1,1)> u$}
    	\State $\nu_t = i$
    \Else 
    	\State $\nu_t = j$
    \EndIf
    \State $t = t+ X$
    \State $R_{t}=R_{t-k} + Y(1-R_{t-k})$
    \State $k =  X$
\Else
    \State Simulate $X \sim Exponential(\lambda_j) $
    \State Simulate $Y \sim Beta(\alpha_j, \beta_j)$
    \State Simulate $u \sim Uniform(0,1)$
    \If{$Q(2,1)> u$}
    	\State $\nu_t = i$
    \Else 
    	\State $\nu_t = j$
    \EndIf
    \State $t = t+ X$
    \State $R_{t} = R_{t-k} + Y(1-R_{t-k})$
    \State $k =  X$
\EndIf
\EndWhile
\end{algorithmic}
\end{algorithm}

In the algorithm we can identify two cases that we can associate to each of the two states of the process. For both, we simulate the times between jumps and the size of the jump. In addition,  a uniform random variable is simulated that allows us  to simulate the permanence or exit of the current state in which the process $R_t$ is.

For example, we assume a two-state matrix $Q$ with the  coefficients defined in \eqref{eq: transition_matrix}. We further assume the distributions of the times between jumps and the distributions of the jumps expressed in Table \ref{table: parameters}. 

Figure \ref{fig: path} shows a set of sample paths simulations to illustrate the behaviour of the process over time, considering the characteristics of the states presented above. 

Similarly, figure \ref{fig: path2} shows the $10,000$ sample path simulation, where the dotted line corresponds to the 5th and 95th percentiles, respectively. Through this procedure we can estimate the expected value and variance of the $R_t$ process by means of a simulation analysis.

\begin{figure}[H]
  \begin{subfigure}{14cm}
    \centering\includegraphics[width=10cm]{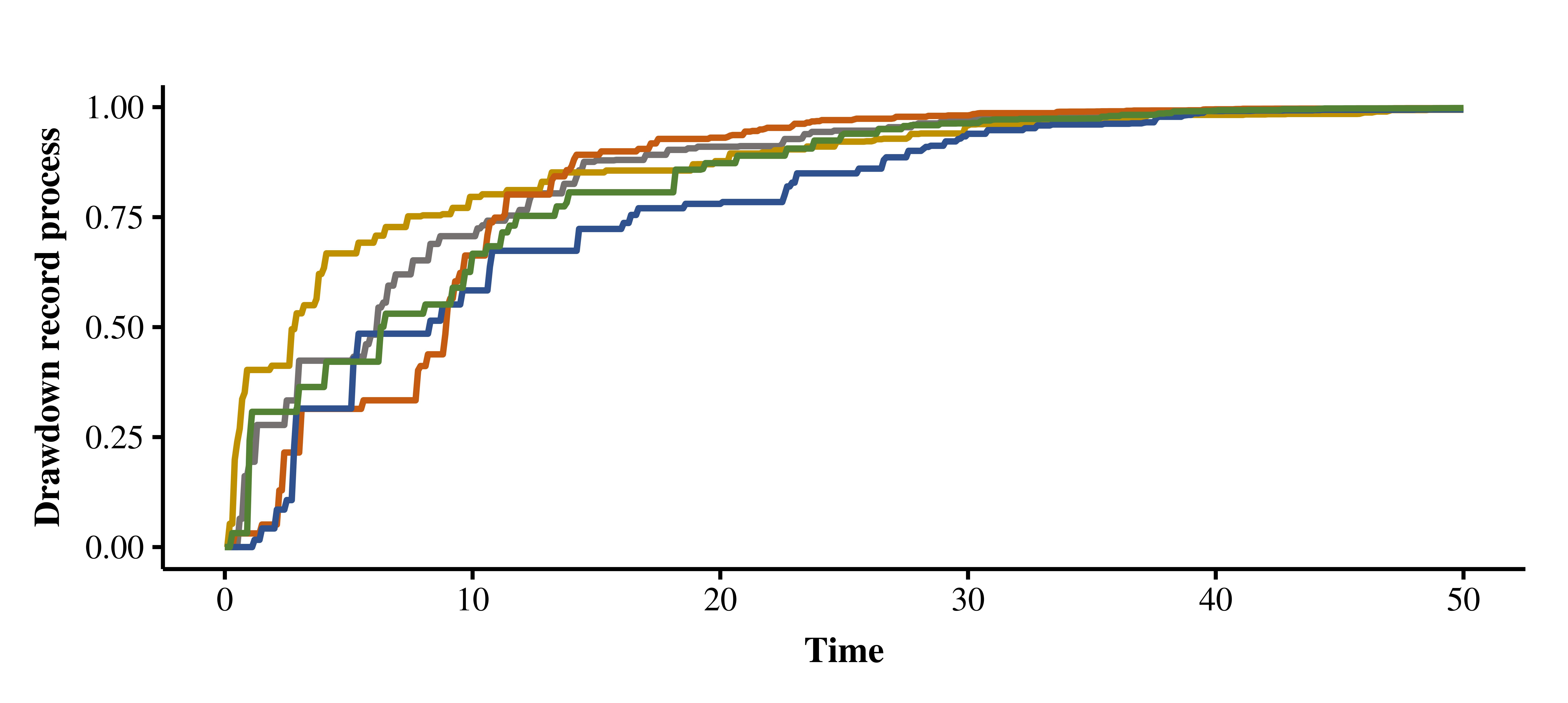}
    \caption{Sample path of $R_t$ process}\label{fig: path}
  \end{subfigure}
  \begin{subfigure}{14cm}
    \centering\includegraphics[width=10cm]{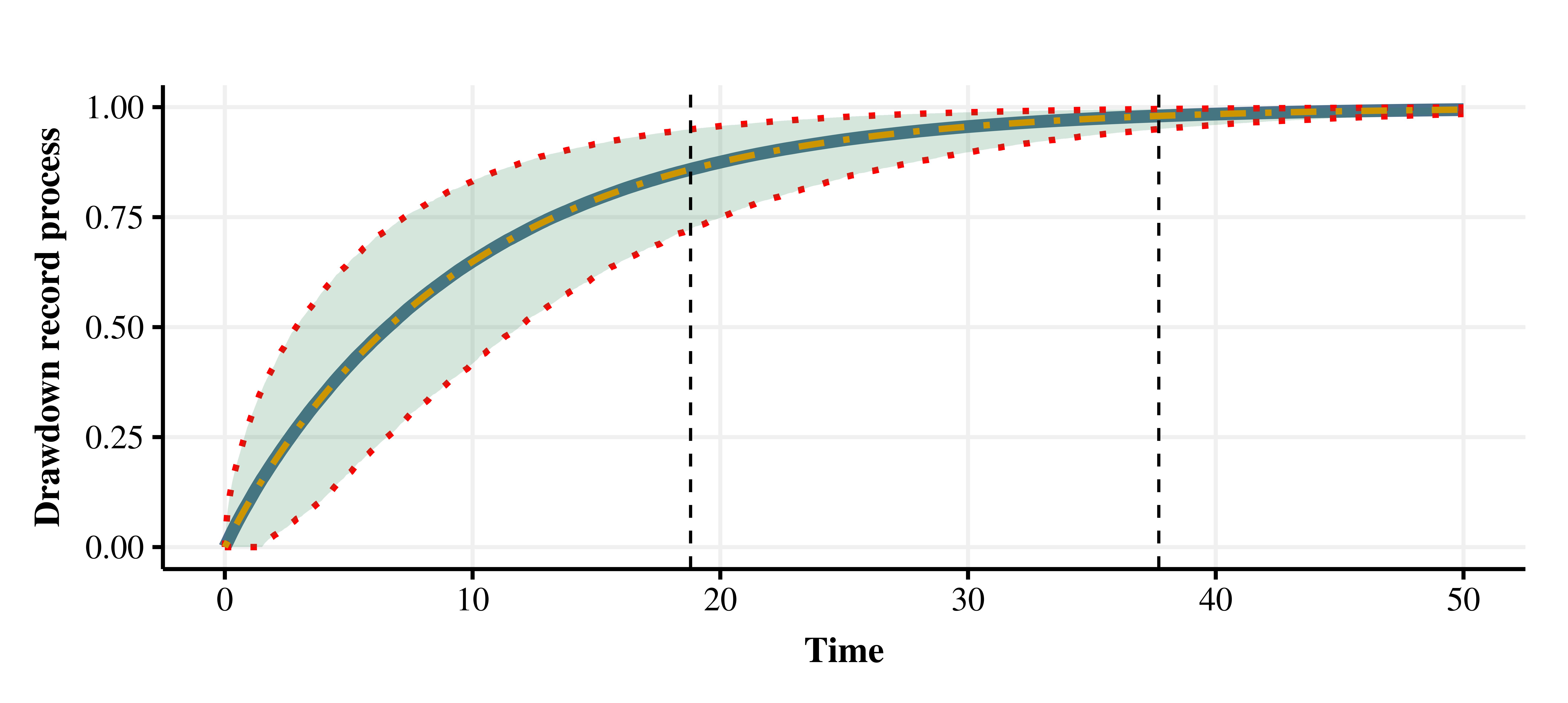}
    \caption{Sample path of $R_t$ process, $10,000$ simulations}\label{fig: path2}
  \end{subfigure} 
  \caption{Sample path of $R_t$ process for a sample of 10 simulate paths and a large sample of 10,000 simulate paths. }
\end{figure}

We can observe that, for this parameterisation, after period $10$ approximately 90\% of the paths exceed the value of $0.5$, after period $20$ approximately 95\%  of them exceed the value of $0.75$, and after period $40$ approximately 95\% of the paths have a difference less than $0.02$ at the limit $1$. 

On the other hand, the greatest variability of the trajectories can be observed between period $5$ and $10$, observing the convergence of the trajectories as time grows as shown in figure \ref{fig: path2}.  

To verify the efficiency of the simulation process, a dot-dashed line corresponding to the analytical mean calculated for those process parameters $R_t$ is placed in the figure. We can see that the computational approximation of the process is very similar to the one obtained through the analytical solution.

This approach makes it possible to characterise the evolution of the main statistics of the $R_t$ process numerically, providing a quick and simple tool for understanding the evolution of the phenomenon. 

Figure \ref{fig: Variance} shows the evolution of the variance of the $R_t$ process, observing a convergence to zero that follows an exponential behaviour, as described in the previous sections.

Similar to Figure \ref{fig: path2}, the analytical solution of the variance by means of a dot-dashed line has been used to show the efficiency of the computational process. Like the mean, the computationally estimated variance is similar to the analytically estimated variance.

\begin{figure}[H]
\centering\includegraphics[width=10cm]{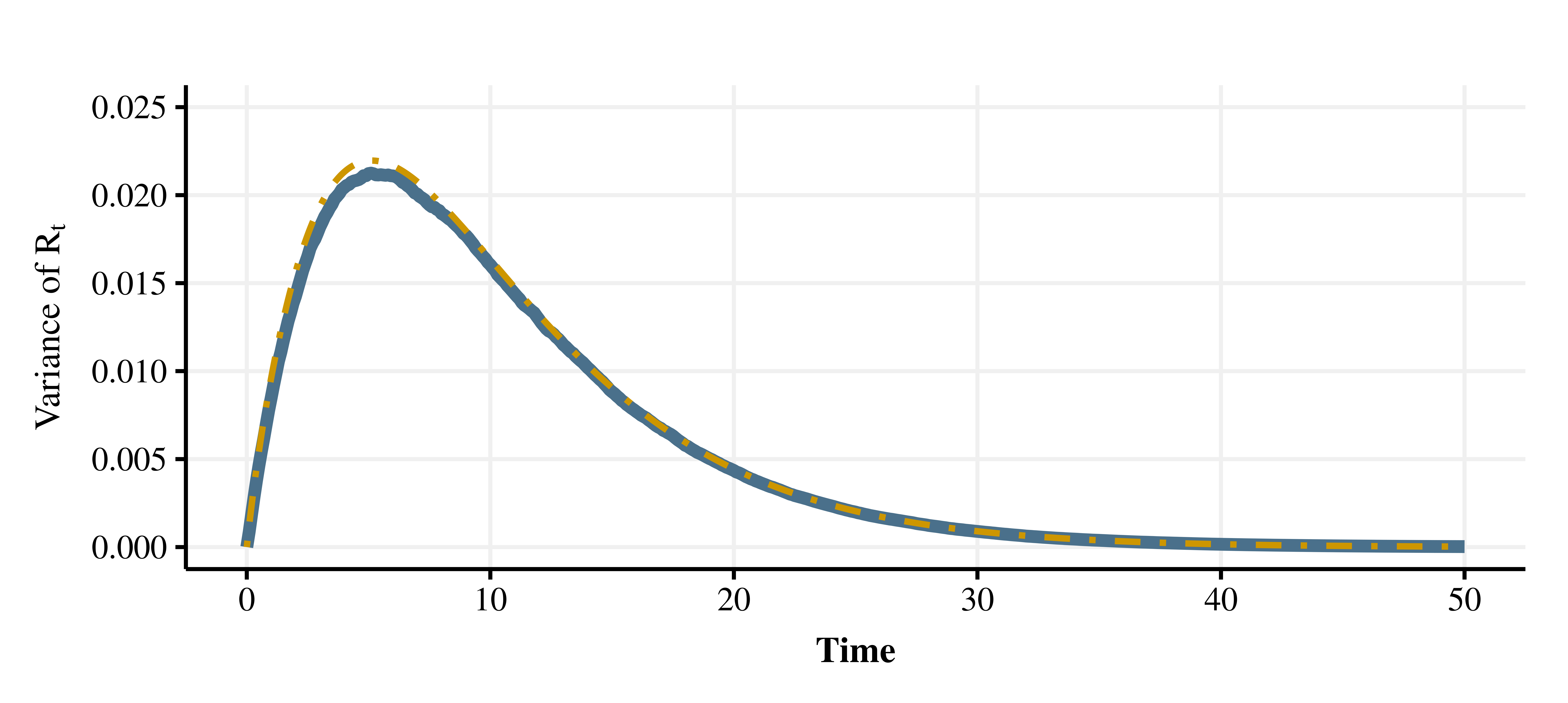}
    \caption{Estimated variance of $R_t$ process, $10,000$ simulations}\label{fig: Variance}

\end{figure}

\subsection{Estimation of the parameters of PDMP record process}

Several studies have developed estimation techniques for PDMP processes. Among which we can mention  \cite{chiquet2009piecewise}, where the authors  investigate some methods to estimate the parameters of the dynamical system, involving  maximum likelihood estimation for the infinitesimal generator of the underlying jump Markov process. In this article we present an estimation based on  the maximum likelihood principles  and distributions on jumps size of the $R_t$ process. 
Algorithm \ref{A2} shows the  way to estimate the parameters of a PDMP - $R_t$ process. 
Based on a specific state, the algorithm estimates the time distribution between jumps and the jump size distribution. We define the function $Par_{exp}(\cdot)$ that allows us to calculate the parameters of the exponential distribution of the vector $x_i$, $x_j$, respectively. Finally,  $Par_{beta}(\cdot)$ allows us to calculate the parameters of the Beta distribution of the vector $y_i$, and $y_j$, respectively.   

We define the function $Like(\cdot)$ that allows us to calculate the maximum likelihood, considering the parameters of $(x_i, y_i)$, and of $(x_j, y_j)$ jointly, obtaining the vector $Likelilihood_i$, and $Likelilihood_j$. Based on the above calculations. At the end,  we define $dif$ which will be a criterion for the convergence of the algorithm.

\begin{algorithm}
\caption{Estimation of the parameters in  $R_t$ process for two states}\label{sim: estim}
\label{A2}
\begin{algorithmic}
\Require Set the initial values of the labels of states of process $R_t$ and the initial state $\nu_0$. Define a guess of the probabilities of the states of the transition matrix $Q$ and initial value $Likelihood_0$. Set a $\delta$ value for measure de convergence of algorithm and a value inicial $dif>\delta$ to start the iterative process. Additionally it sets the time counter $k$ and inicial time $t$ to zero and $s=1$, $g=1$, $r=1$.
\While{$dif \geq \delta$}
\While{$t \leq T$}
\If{$E_r = i$}
    \State $x_{i,s} = x_r$
    \State $y_{i,s} = y_r$
    \State $s = s+ 1$
\Else
    \State $x_{j,g} = x_r$
    \State $y_{j,g} = y_r$
    \State $g = g+ 1$
\EndIf
\State $t = t+ x_r$
\State $r = r+ 1$
\EndWhile
\State $Likelihood_{i} =Like(Par_{exp}(x_i),Par_{beta}(y_i),x,y)$  
\State $Likelihood_{j} =Like(Par_{exp}(x_j),Par_{beta}(y_j),x,y)$
\State $dif = Likelihood_0 -\left(\sum(Likelihood_i) +\sum(Likelihood_j) \right) $
\State $Likelihood_0 = \sum(Likelihood_i) +\sum(Likelihood_j) $
\EndWhile
\end{algorithmic}
\end{algorithm}

Once the distributions have been estimated and the labels of the specific state to which it belongs have been assigned, it is necessary to estimate the transition matrix $Q$. This estimation of the transition matrix will input a new estimation process of Algorithm \ref{A2}, repeating the process until it converges. The following Algorithm \ref{A3} allows to generate the classification of the state labels. The first cycle generates the new classification of the labels, given the maximum likelihood calculations defined in Algorithm \ref{A2}. The second cycle allows the estimation of the state matrix $Q$. Consider that the algorithm can be repeated until there are no variations of the state matrices and the calculated likelihoods.

\begin{algorithm}
\caption{Estimation of the states in the PDMP- $R_t$ process with two states}\label{A3}
\begin{algorithmic}
\Require Set the initial values of the labels of states of process $R_t$ and the initial state $\nu_0$. Define a guess of the probabilities of the states of the transition matrix $Q$ and initial value $Likelihood_0$. Set a $\delta$ value for measure de convergence of algorithm and a value inicial $dif>\delta$ to start the iterative process. Additionally it sets the time counter $k$ and inicial time $t$ to zero and $s=1$, $g=1$, $r=1$.
\For{$r = 1:R$}
\If{$Likelihood_{i,r} > Likelihood_{j,r}$}
    \State $E_r = i$
\Else
    \State $E_r = j$
\EndIf

\EndFor

\For{$r = 1:R$}
\If{$E_r = i$}
    \State $C_i = C_i + 1$
    \If{$E_{r+1} = i$}
    	\State $C_{i,i} = C_{i,i} + 1$
	\EndIf
\Else
    \State $C_j = C_j + 1$
    \If{$E_{r+1} = j$}
    	\State $C_{j,j} = C_{j,j} + 1$
	\EndIf
\EndIf
\EndFor
\State $Q_{i,i} = C_{i,i}/C_i$
\State $Q_{j,j} = C_{j,j}/C_j$
\State $Q_{i,j} = 1-Q_{i,i}$
\State $Q_{j,i} = 1-Q_{j,j}$
\end{algorithmic}
\end{algorithm}

\section{Financial applications}

Figure \ref{fig: SP500} shows the evolution of the S$\&$P500 index which represents the $500$ most traded companies on the New York Stock Exchange from 1 January 1950 to 31 December 2019, public data extracted from https://finance.yahoo.com/. The S$\&$P500 is used as a proxy indicator for economic performance, which assumes that there is full transparency in the investment positions of each economic agent.

\begin{figure}[H]
\centering
  \centering\includegraphics[width=10cm]{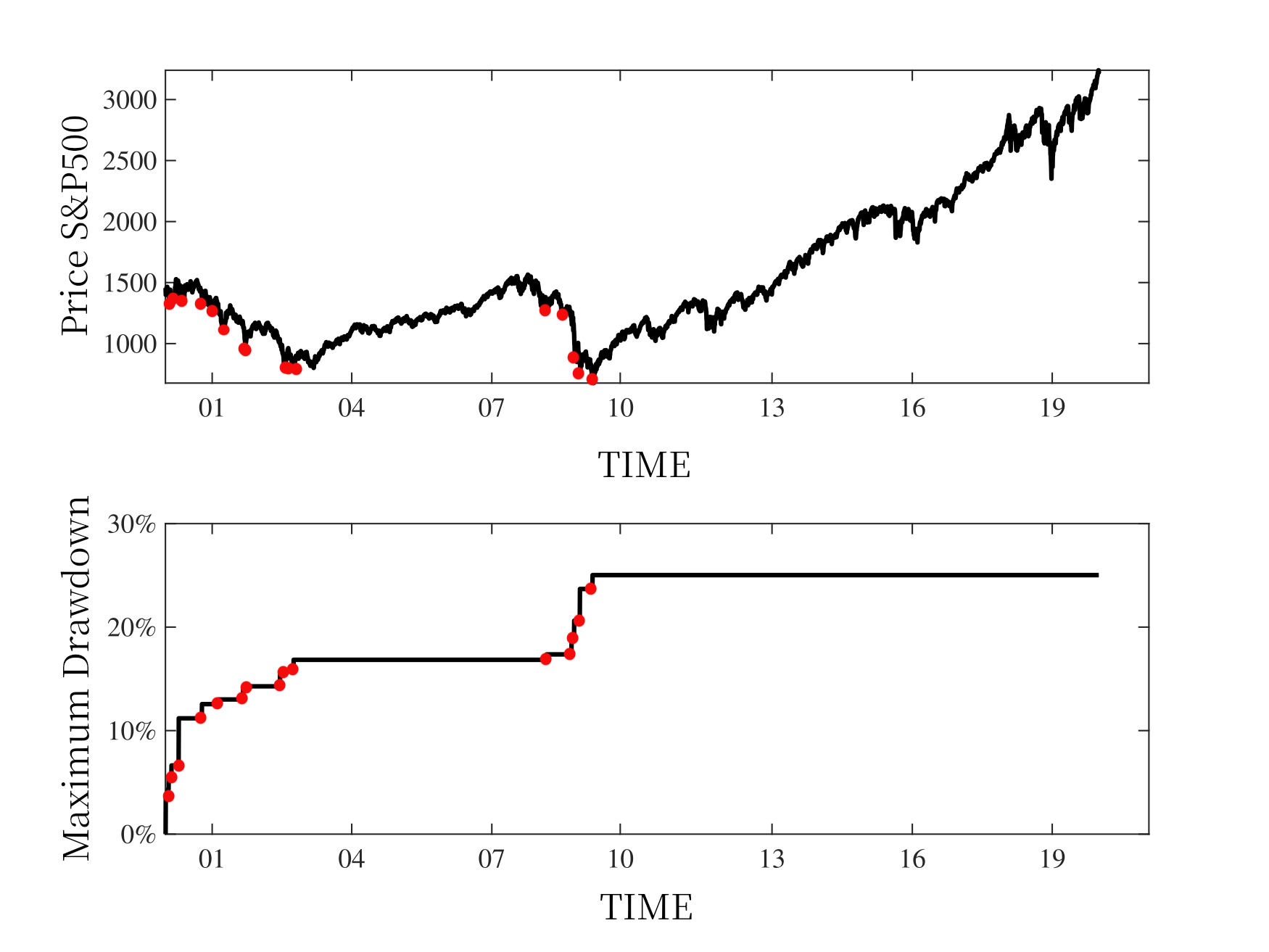} 
\caption{S\&P500 index and Maximum Drawdown.}\label{fig: SP500}
\end{figure}

When analysing the figure \ref{fig: SP500}, we can observe periods of time in which the observation of a new record is more recurrent, while in other periods the frequency of new records is lower. The red circles highlight the moment at which a new record is set. The lower time series shows the value of the record for a period t, and the size of the jump corresponds to the difference of the record with the immediately preceding period. This suggests the existence of at least 2 different states of nature to observe a new record in this time series.

 




Using the Algorithms \ref{A2} and \ref{A3} on the data obtained from the maximum drawdown process for the S\&P500 time series showing the evolution of records, we have the following results:
\begin{table}[h]
\centering
\begin{tabular}{lcl}
\cline{2-3}
                                    & \multicolumn{2}{c}{Distribution}                                   \\
                                    & $\nu_t= i$                        & \multicolumn{1}{c}{$\nu_t= j$} \\ \hline
\multicolumn{1}{r}{$X \sim$} & $Exponential(0.47)$                 & $Exponential(5.4\times 10^{-4})$               \\
\multicolumn{1}{r}{$\E[X] =$}     & \multicolumn{1}{l}{$2.095$} & $1819$       \\         
\multicolumn{1}{r}{$\rho \sim$}     & \multicolumn{1}{l}{$Beta(1.83, 145.90)$} & $Beta(0.77, 47.86)$                  \\
\multicolumn{1}{r}{$\E[\rho]=$}     & \multicolumn{1}{l}{$0.012$} & $0.015$                  \\ \hline
\end{tabular}
\caption{Parameter Inference of the  S\&P500 -$R_t$ process.}
\end{table}
which represent the parameter values of the estimated distributions for each of the two states present in the $R_t$ process. In addition, by means of Algorithm xxx we can estimate the transition matrix $Q$. These results allow us to estimate the long-term behaviour of the time series shown in Figure \ref{fig: path_empirica} and \ref{fig: variance_empirica}.
\begin{equation}
Q= \left[
\begin{matrix}
q_{1,1} &  q_{1,2}\\
q_{2,1} &  q_{2,2}
\end{matrix}\right] = \left[
\begin{matrix}
0.883 &  0.117\\
0.750 &  0.250
\end{matrix}
\right],
\end{equation}

With all the parameters obtained according to Algorithms 2 and 3 we can numerically calculate the mean and the variance of the estimated processes which we will call $\hat{R}_r$. In figures \ref{fig: path_empirica} and \ref{fig: variance_empirica} we show the estimation of the mean and variance of the S\&P500 - $R_t$ process.

\begin{figure}[H]
\centering\includegraphics[width=10cm]{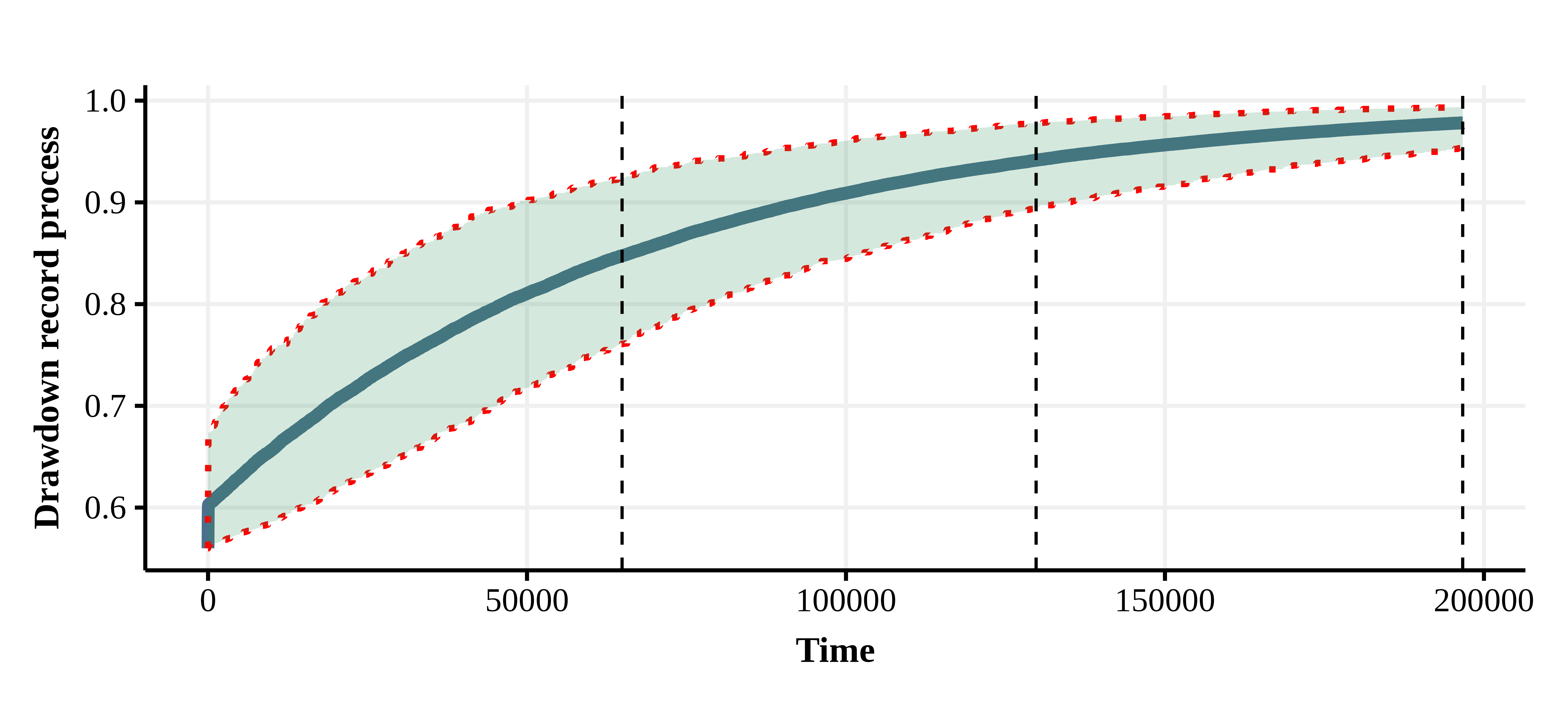}
    \caption{Sample path of $\hat{R}_t$ process, 10,000 simulations}\label{fig: path_empirica}
\end{figure}

\begin{figure}[H]
\centering\includegraphics[width=10cm]{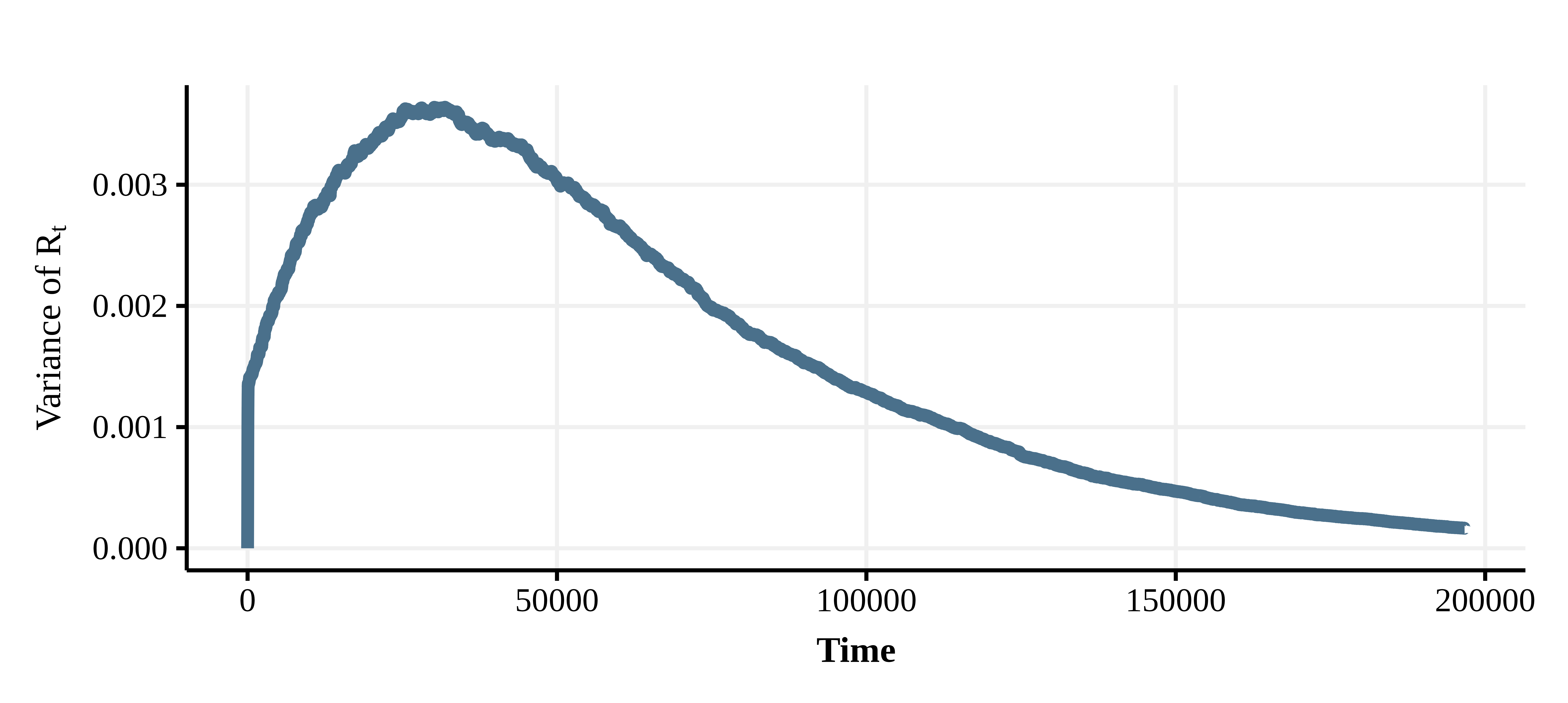}
    \caption{Variance of $\hat{R}_t$ process, 10,000 simulations}\label{fig: variance_empirica}
\end{figure}


\section{Conclusions}

The maximum drawdown is a widely employed risk management strategy within the domain of capital markets, whereby the maximum observed decline over a specified time interval can be identified. This research establishes a link between a piecewise deterministic Markov process and record theory, specifically in relation to the maximum drawdown. 

This research presents a statistical analysis of a record process utilising a maximum drawdown approach. The characteristics that drive the maximum drawdown process, such as the mean and variance, were defined, thereby providing analytical tools for defining the characteristics of the limit process, in which a finite number of states are considered. This latter process has numerous applications in the context of financial risk management.

In addition, computational estimation techniques were developed, allowing the estimation of the parameters governing a given time series, such as hop occurrence times, process hop sizes, and the state matrix, which makes it possible to simulate processes and build useful applications.

The expected future work includes the application of a generalisation of this work, considering Mittag-Leffler type distributions as the distribution of the time between jumps, which requires modifying the differential operators. In addition to considering modifications to the process, such as the possibility of the jump range being greater than one, which has other theoretical and practical implications.

\section*{Acknowledgments}

Soledad Torres was partially supported by Fondecyt project number 1230807 and Matham- sud SMILE AMSUD230032,Proyecto ECOS210037 (C21E07), Mathamsud AMSUD210023 and Fondecyt Regular No. 1221373.

Lisandro Fermín was partially supported by by Fondecyt project number 1230807, MATH-AMSUD 23-MATH-12, MathAmSud Tomcat 22-math-10, and the project Labex MME-DII-2024-2-0000000018 (ANR11-LBX-0023-01).


\bibliographystyle{elsarticle-harv} 
\bibliography{biblio.bib}

\end{document}